\documentclass[draftcls,journal,onecolumn,11pt]{IEEEtran}

\IEEEoverridecommandlockouts                              
\usepackage{graphicx}
\usepackage{amsmath}
\usepackage{amssymb}
\usepackage[font=footnotesize]{caption} 
\usepackage{subfig} 
\usepackage{cite} 
\usepackage{color}
\usepackage{algorithm} 
\usepackage{algpseudocode} 
\usepackage{enumerate}

\usepackage{tikz}
\usetikzlibrary{calc} 

\newtheorem{theorem}{Theorem}
\newtheorem{lemma}{Lemma}
\newtheorem{assumption}{Assumption}
\newtheorem{remark}{Remark}

\graphicspath{{figures/}}

\newcommand{\T}{\mathrm{T}}

\newcommand{\Null}{\mathrm{Null}}

\newcommand{\one}{\mathbf{1}}
\newcommand{\rank}{\mathrm{rank}}
\newcommand{\sgn}{\mspace{2mu}\mathrm{sgn}}

\newcommand{\Ln}{\mspace{2mu}\mathrm{Ln}}

\begin{document}

\title{\LARGE \bf
Finite-time Stabilization of Circular Formations \\ using Bearing-only Measurements
}

\author{Shiyu Zhao,
        Feng Lin,
        Kemao Peng,
        Ben M. Chen
        and Tong H. Lee 
\thanks{S. Zhao, B. M. Chen and T. H. Lee are with the Department of Electrical and Computer Engineering, National University of Singapore, Singapore 117576, Singapore
    {\tt\small \{shiyuzhao, bmchen, eleleeth\}@nus.edu.sg}}
\thanks{F. Lin and K. Peng are with the Temasek Laboratories, National University of Singapore, Singapore 117456, Singapore
    {\tt\small \{linfeng, kmpeng\}@nus.edu.sg}}
}

\maketitle

\begin{abstract}
This paper studies decentralized formation control of multiple vehicles when each vehicle can only measure the local bearings of their neighbors by using bearing-only sensors.
Since the inter-vehicle distance cannot be measured, the target formation involves no distance constraints.
More specifically, the target formation considered in this paper is an angle-constrained circular formation, where each vehicle has exactly two neighbors and the angle at each vehicle subtended by its two neighbors is pre-specified.
To stabilize the target formation, we propose a discontinuous control law that only requires the sign information of the angle errors.
Due to the discontinuity of the proposed control law, the stability of the closed-loop system is analyzed by employing a locally Lipschitz Lyapunov function and nonsmooth analysis tools.
We prove that the target formation is locally finite-time stable with collision avoidance guaranteed.
The evolution of the vehicle positions in the plane is also characterized.
\end{abstract}

\begin{IEEEkeywords}
Bearing-only measurement; discontinuous dynamic system; finite-time stability; formation stabilization; Lyapunov function
\end{IEEEkeywords}

\IEEEpeerreviewmaketitle

\section{Introduction}
Formation control of multiple vehicles has been studied extensively in the last decade.
Inter-vehicle information exchange is a necessary condition for distributed formation control.
It is commonly assumed that each vehicle can obtain the relative \emph{position} information of their neighbor vehicles via, for example, wireless communication.
It is notable that position information essentially consists of two kinds of partial information: \emph{range} and \emph{bearing}.
In recent years, formation control using range-only \cite{cao2011,cao2012} or bearing-only \cite{bishopSCL,bishopconf2011relax,bishopconf2011quad,ErenIJC,nimaconf,nimaTR} measurements has become an active research area.
Up to now, many problems are still unsolved in this area.
In this paper, we will particularly study formation control using bearing-only measurements.
We assume that each vehicle is only able to measure the bearings of their neighbors by using, for example, monocular or omnidirectional cameras, which are inherently bearing-only sensors and have been applied in many control-related tasks.
Vision-based formation control \cite{nimaconf,nimaTR,TRA2002Vision} could be a potential application of our work.

A number of interesting problems arise when only bearing measurements are available for formation control.
One important problem is how to utilize these bearing-only measurements.
There are generally two possible schemes.
In the first scheme, each vehicle can track their neighbors using the bearing-only measurements such that the positions of their neighbors can be estimated and then used for formation control.
There exist various bearing-only target tracking algorithms (see, for example, \cite{1984TAC}).
But it should be noted that bearing-only target tracking requires certain observability conditions \cite{1981AES}.
As a trivial example, suppose two vehicles can measure the bearings of each other but have no relative motion.
Then it would be impossible for them to estimate their inter-vehicle distance using bearing-only tracking.
As a result, if the relative positions of the vehicles are supposed to be fixed in the target formation, the observability problem would become severe and then the first scheme is inapplicable.
In order to apply the first scheme anyway, one may adopt a stop-and-go strategy like the one proposed in \cite{cao2011}.

In this work we will focus on the second scheme, in which the formation control law is directly implemented based on the bearing-only measurements.
No vehicle position estimation is involved.
If merely bearing measurements are used for feedback control, the inter-vehicle distance in the formation would be uncontrollable.
As a result, any constraints involving inter-vehicle distance cannot be specified in the target formation.
It is only possible to specify the bearings of the edges that connect vehicles \cite{ErenIJC,bishopconf2011rigid} or the angles at each vehicle subtended by their neighbors \cite{bishopSCL,bishopconf2011relax,bishopconf2011quad}.
If the target formation is constrained by edge bearings, global bearing measurements are required.
That means the bearing measurements of different vehicles should be taken in one global coordinate frame.
As a comparison, if the target formation is constrained by angles, each vehicle may measure the bearings of their neighbors in their local coordinate frames.
In this paper, we will particularly study the angle-constrained case.
It should be noted that the realization of an angle-constrained target formation would not be unique.
More specifically, the orientation or translation of the formation in the plane, or the scale of the formation is not unique.
In our work, we make no parallel rigid assumptions \cite{eren2003,ErenIJC,bishopconf2011rigid}.
Hence the shape of the target formation might not be unique either.

Collision avoidance is an important issue in formation control problems.
This issue is especially important and also challenging to analyze in formation control using bearing-only measurements because the inter-vehicle distance cannot be measured.
In \cite{bishopSCL}, a bearing-only control law is proposed to globally stabilize a triangle formation of three vehicles.
It is proved that collision can be avoided naturally by the proposed control law.
In our previous work \cite{zhaoArxiv}, we extended the work in \cite{bishopSCL} and proposed a bearing-only control law to stabilize circular formations of an arbitrary number of vehicles.
By employing Lyapunov approaches, we proved that collision avoidance can be ensured if the initial angle errors are sufficiently small.
A similar idea will be adopted in this paper to tackle the collision avoidance issue. We will prove that the formation can be stabilized within finite time before any vehicles could possibly collide.

In this paper, we study distributed formation stabilization using local bearing-only measurements.
The target formation considered in this paper is an angle-constrained circular formation, where each vehicle has exactly two neighbors.
The underlying information flow is described by an undirected circular graph with fixed topology.
The angle at each vehicle subtended by its two neighbors is constrained in the target formation.
We propose a distributed discontinuous control law to stabilize the target formation.
The proposed control law only requires sign information of the angle errors and is able to stabilize the target formation in finite time.
Finite-time control has attracted much attention in recent years \cite{XielihuaAutomatica,XielihuaASME,LinZongliACC,JiangZhongpingSIAM,HuangJieSCL,corte2006automatica,wanglong2009finiteformation}, to name a few.
Besides fast convergence, finite-time control can also bring benefits such as disturbance rejection and robustness against uncertainties \cite{bernstein2000}.
Due to the discontinuity of the proposed control law, we employ a locally Lipschitz Lyapunov function and nonsmooth analysis tools \cite{bookFilippov,bookClarke,1987Paden,1999Bacciotti,corte2005SIAM,corte2006automatica,corte2008CSM} to prove the finite-time stability of the closed-loop system.
It is also proved that collision avoidance can be guaranteed if the initial angle errors are sufficiently small.

The paper is organized as follows.
Preliminaries regarding graph theory and nonsmooth analysis are introduced in Section \ref{section_preliminary}.
The formation control problem is formulated in Section \ref{section_problemstatement}.
The formation stability and behavior are analyzed in Section \ref{section_stabilityAnalysis}.
Section \ref{section_simulation} presents simulation results.
Conclusions are drawn in Section \ref{section_conclusion}.

\section{Notations and Preliminaries}\label{section_preliminary}

\subsection{Notations}
Given a symmetric positive semi-definite matrix $A\in\mathbb{R}^{n\times n}$, the eigenvalues of $A$ are denoted as $0\le\lambda_1(A)\le\lambda_2(A)\le\dots\le\lambda_n(A)$.
Let $\one=[1,\dots,1]^{\T}\in\mathbb{R}^n$, and $I$ be the identity matrix with appropriate dimensions.
Denote $|\cdot|$ as the absolute value of a real number, and $\|\cdot\|$ as the Euclidean norm of a vector.
Denote $\Null(\cdot)$ as the right null space of a matrix.
Let $[\,\cdot\,]_{ij}$ be the entry at the $i$th row and $j$th column of a matrix, and $[\,\cdot\,]_i$ be the $i$th entry of a vector.
Given a set $S$, denote $\overline{S}$ as its closure.
For any angle $\alpha\in\mathbb{R}$,
\begin{align}\label{eq_rotation_matrix}
    R(\alpha)=\left[
                \begin{array}{cc}
                  \cos \alpha & -\sin\alpha \\
                  \sin\alpha & \cos\alpha \\
                \end{array}
              \right]\in\mathbb{R}^{2\times 2}
\end{align}
is a rotation matrix satisfying $R^{-1}(\alpha)=R^{\T}(\alpha)=R(-\alpha)$.
Geometrically, $R(\alpha)$ rotates a vector in $\mathbb{R}^2$ counterclockwise through an angle $\alpha$ about the origin.

\subsection{Graph Theory}
A graph $\mathcal{G}=(\mathcal{V},\mathcal{E})$ consists of a vertex set $\mathcal{V}=\{1,\dots,n\}$ and an edge set $\mathcal{E}\subseteq \mathcal{V} \times \mathcal{V}$.
If $(i,j)\in \mathcal{E}$, then $i$ and $j$ are called to be adjacent.
The set of neighbors of vertex $i$ is denoted as $\mathcal{N}_i=\{j \in \mathcal{V} \ |\ (i,j)\in \mathcal{E}\}$.
A graph is undirected if each $(i,j)\in \mathcal{E}$ implies $(j,i)\in \mathcal{E}$, otherwise the graph is directed.
A path from $i$ to $j$ in a graph is a sequence of distinct nodes starting with $i$
and ending with $j$ such that consecutive vertices are adjacent.
If there is a path between any two vertices of graph $\mathcal{G}$, then $\mathcal{G}$ is said to be connected.
An undirected circular graph is a connected graph where every vertex has exactly two neighbors.

An incidence matrix of a directed graph is a matrix $E$ with rows indexed by
edges and columns indexed by vertices\footnote{In some literature such as \cite{graphbook}, the rows of an incidence matrix are indexed by
vertices and the columns are indexed by edges.}.
Suppose $(j,k)$ is the $i$th edge.
Then the entry of $E$ in the $i$th row and $k$th column is $1$, the one in the $i$th row and $j$th column is $-1$, and the others in the $i$th row are zero.
By definition, we have $E\one=0$.
If a graph is connected, the corresponding $E$ has rank $n-1$ (see \cite[Theorem 8.3.1]{graphbook}).
Then $\Null(E)=\mathrm{span}\{\one\}$.

\subsection{Nonsmooth Stability Analysis}\label{subsection_nonsmoothanalysis}
Next we introduce some useful concepts and facts regarding discontinuous dynamic systems \cite{bookFilippov,bookClarke,1987Paden,1999Bacciotti,corte2005SIAM,corte2006automatica,corte2008CSM}.

\subsubsection{Filippov Differential Inclusion}

Consider the dynamic system
\begin{align}\label{eq_generalSystem}
    \dot{x}(t)=f\left(x(t)\right),
\end{align}
where $f: \mathbb{R}^n\rightarrow \mathbb{R}^n$ is a measurable and essentially locally bounded function.
The Filippov differential inclusion \cite{bookFilippov} associated with the system \eqref{eq_generalSystem} is
\begin{align}\label{eq_FilippovInclusion}
    \dot{x}\in\mathcal{F}[f](x),
\end{align}
where $\mathcal{F}[f]: \mathbb{R}^n\rightarrow 2^{\mathbb{R}^n}$ is defined by
\begin{align}\label{eq_FilippovDefinition}
    \mathcal{F}[f](x)= \bigcap_{r>0}\bigcap_{\mu(S)=0}\overline{\mathrm{co}}\left\{f\left(B(x,r)\setminus S\right)\right\}.
\end{align}
In \eqref{eq_FilippovDefinition}, $\overline{\mathrm{co}}$ denotes convex closure, $B(x,r)$ denotes the open ball centered at $x$ with radius $r>0$, and $\mu(S)=0$ means that the Lebesgue measure of the set $S$ is zero.
The set-valued map $\mathcal{F}[f]$ associates each point $x$ with a set.
Note $\mathcal{F}[f](x)$ is multiple valued only if $f(x)$ is discontinuous at $x$.

A Filippov solution of \eqref{eq_generalSystem} on $[0,t_1]\subset\mathbb{R}$ is defined as an absolutely continuous function $x: [0,t_1]\rightarrow \mathbb{R}^n$ that satisfies \eqref{eq_FilippovInclusion} for almost all $t\in[0,t_1]$.
If $f(x)$ is measurable and essentially locally bounded, the existence of Filippov solutions can be guaranteed \cite[Lemma 2.5]{corte2005SIAM} \cite[Proposition 3]{corte2008CSM} though the uniqueness cannot.
The interested reader is referred to \cite[p. 52]{corte2008CSM} for the uniqueness conditions of Filippov solutions.
A solution is called maximal if it cannot be extended forward in time.
A set $\Omega$ is said to be weakly invariant (respectively strongly invariant) for \eqref{eq_generalSystem}, if for each $x(0)\in\Omega$, $\Omega$ contains at least one maximal solution (respectively all maximal solutions) of \eqref{eq_generalSystem}.

\subsubsection{Generalized Gradient}

Suppose $V: \mathbb{R}^n\rightarrow \mathbb{R}$ is a locally Lipschitz function.
If $V(x)$ is differentiable at $x$, denote $\nabla V(x)$ as the gradient of $V(x)$ with respect to $x$.
Let $M_V$ be the set where $V(x)$ fails to be differentiable.
The generalized gradient \cite{bookClarke,corte2005SIAM,corte2008CSM} of $V(x)$ is defined as
\begin{align*}
    \partial V(x)=\mathrm{co}\left\{\lim_{i\rightarrow +\infty} \nabla V(x_i) \ |\  x_i\rightarrow x,\, x_i\notin S \cup M_V \right\},
\end{align*}
where $\mathrm{co}$ denotes convex hull and $S$ is an arbitrary set of Lebesgue measure zero.
The generalized gradient is a set-valued map. 
If $V(x)$ is continuously differentiable at $x$, then $\partial V(x)=\{\nabla V(x)\}$.

Given any set $S\subseteq\mathbb{R}^n$, let $\Ln: 2^{\mathbb{R}^n}\rightarrow2^{\mathbb{R}^n}$ be the set-valued map that associates $S$ with the set of least-norm elements of $\overline{S}$.
If $S$ is convex, $\Ln(S)$ is singleton.
In this paper, we only apply $\Ln$ to generalized gradients which are always convex.
For a locally Lipschitz function $V(x)$, $\Ln(\partial V): \mathbb{R}^n\rightarrow\mathbb{R}^n$ is called the generalized gradient vector field.
The following fact \cite[Proposition 8]{corte2008CSM}
\begin{align}\label{eq_FilippovFact}
    \mathcal{F}\left[\Ln\left(\partial V(x)\right)\right]=\partial V(x)
\end{align}
will be very useful in our work.
A point $x$ is called a critical point if $0\in\partial V(x)$.
For a critical point $x$, it is obvious that $\Ln(\partial V(x))=\{0\}$.

\subsubsection{Set-valued Lie Derivative}
The evolution of a locally Lipschitz function $V(x)$ along the solutions to the differential inclusion $\dot{x}\in\mathcal{F}[f](x)$ can be characterized by the set-valued Lie derivative \cite{1999Bacciotti,corte2005SIAM,corte2008CSM}, which is defined by
\begin{align*}
    \widetilde{\mathcal{L}}_{\mathcal{F}}V(x)=\left\{ \ell\in\mathbb{R} \ |\  \exists \xi\in\mathcal{F}[f](x),\  \forall \zeta\in\partial V(x),\  \xi^{\T}\zeta=\ell \right\}.
\end{align*}
With a slight abuse of notation, we also denote $\widetilde{\mathcal{L}}_{f}V(x)=\widetilde{\mathcal{L}}_{\mathcal{F}}V(x)$.
The set-valued Lie derivative may be empty.
When $\widetilde{\mathcal{L}}_{\mathcal{F}}V(x)=\emptyset$, we take $\max \widetilde{\mathcal{L}}_{\mathcal{F}}V(x)=-\infty$ (see \cite{1999Bacciotti,corte2005SIAM,corte2008CSM}).

A function $V: \mathbb{R}^n\rightarrow \mathbb{R}$ is called regular \cite[p. 57]{corte2008CSM} at $x$ if the right directional derivative of $V(x)$ at $x$ exists and coincides with the generalized directional derivative of $V(x)$ at $x$.
Note a locally Lipschitz and convex function is regular.
The following two lemmas are useful for proving the stability of discontinuous systems using nonsmooth Lyapunov functions.
The next result can be found in \cite{1994Paden,1999Bacciotti,corte2005SIAM,corte2006automatica}.
\begin{lemma}\label{lemma_invariancePrinciple}
    Let $V: \mathbb{R}^n\rightarrow \mathbb{R}$ be a locally Lipschitz and regular function.
    Suppose the initial state is $x_0$ and let $\Omega(x_0)$ be the connected component of $\{x\in\mathbb{R}^n \ |\  V(x)\le V(x_0)\}$ containing $x_0$.
    Assume the set $\Omega(x_0)$ is bounded.
    If $\max \widetilde{\mathcal{L}}_{f}V(x)\le 0$ or $\widetilde{\mathcal{L}}_{f}V(x)=\emptyset$ for all $x\in\Omega(x_0)$,
    then $\Omega(x_0)$ is strongly invariant for \eqref{eq_generalSystem}.
    Let
    \begin{align}\label{eq_Z_f_V}
        Z_{f,V}=\{x\in\mathbb{R}^n \ |\  0\in\widetilde{\mathcal{L}}_{f}V(x)\}.
    \end{align}
    Then any solution of \eqref{eq_generalSystem} starting from $x_0$ converges to the largest weakly invariant set $M$ contained in $\overline{Z}_{f,V}\cap \Omega(x_0)$.
    Furthermore, if the set $M$ is a finite collection of points, then the limit of all solutions starting from $x_0$ exists and equals one of them.
\end{lemma}

The next result can be found in \cite{1987Paden,corte2005SIAM,corte2006automatica}.
\begin{lemma}\label{lamma_finite_time}
    Let $V: \mathbb{R}^n\rightarrow \mathbb{R}$ be a locally Lipschitz and regular function.
    Suppose the initial state is $x_0$ and let $S$ be a compact and strongly invariant set for \eqref{eq_generalSystem}.
    If $\max \widetilde{\mathcal{L}}_{f}V(x)\le -\kappa<0$ almost everywhere on $S\setminus Z_{f,V}$,
    then any solution of \eqref{eq_generalSystem} starting at $x_0\in S$ reaches $Z_{f,V}\cap S$ in finite time.
    The convergence time is upper bounded by $\left(V(x_0)-\min_{x\in S} V(x)\right)/\kappa$.
\end{lemma}

\section{Problem Statement}\label{section_problemstatement}
In this section, we first describe the formation control problem that we are going to solve.
Then we present our proposed control law and derive the closed-loop system dynamics.

\subsection{Angle-constrained Circular Formation}

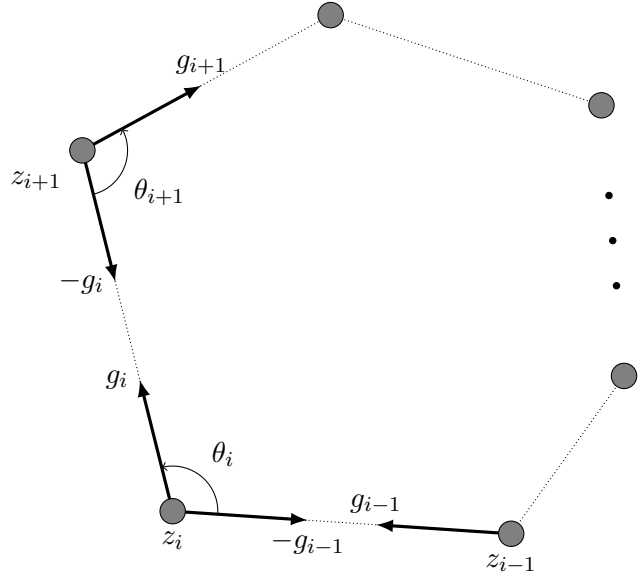
\begin{figure}
  \centering
    \def\myscale{0.6}
    \begin{tikzpicture}[scale=\myscale]
            \coordinate (zim2) at (10,3);
            \coordinate (zim1) at (7.5,-0.5);
            \coordinate (zi)   at (0,0);
            \coordinate (zip1) at (-2,8);
            \coordinate (zip2) at (3.5,11);
            \coordinate (zip3) at (9.5,9);
            \draw [densely dotted, thin] (zim2)--(zim1)--(zi)--(zip1)--(zip2)--(zip3);
            \def\unitlength{3cm}
            \draw[->, >=latex, very thick] (zi) -- ($(zi)!\unitlength!(zim1)$) node [below] {$-g_{i-1}$};
            \draw[->, >=latex, very thick] (zi) -- ($(zi)!\unitlength!(zip1)$) node [left] {$g_{i}$};
            \draw[->, >=latex, very thick] (zim1) -- ($(zim1)!\unitlength!(zi)$) node [above] {$g_{i-1}$};
            \draw[->, >=latex, very thick] (zip1) -- ($(zip1)!\unitlength!(zi)$) node [left] {$-g_{i}$};
            \draw[->, >=latex, very thick] (zip1) -- ($(zip1)!\unitlength!(zip2)$) node [above] {$g_{i+1}$};
            \def\angleRadius{1cm}
            \draw[->] let \p1=(zi), \p2=(zim1), \p3=(zip1), \n1={atan2(\x2-\x1,\y2-\y1)}, \n2={atan2(\x3-\x1,\y3-\y1)}  in
                    ($(\p1)!\angleRadius!(\p2)$) arc (\n1:\n2:\angleRadius);
            \draw[] let \p1=(zi), \p2=(zim1), \p3=(zip1), \n1={atan2(\x2-\x1,\y2-\y1)}, \n2={atan2(\x3-\x1,\y3-\y1)} in
                    (\p1)+(\n1/2+\n2/2:\angleRadius) node[above right] {$\theta_i$};
            \draw[->] let \p1=(zip1), \p2=(zi), \p3=(zip2), \n1={atan2(\x2-\x1,\y2-\y1)}, \n2={atan2(\x3-\x1,\y3-\y1)}  in
                    ($(\p1)!\angleRadius!(\p2)$) arc (\n1:\n2:\angleRadius);
            \draw[] let \p1=(zip1), \p2=(zi), \p3=(zip2), \n1={atan2(\x2-\x1,\y2-\y1)}, \n2={atan2(\x3-\x1,\y3-\y1)} in
                    (\p1)+(\n1/2+\n2/2:\angleRadius) node[below right] {$\theta_{i+1}$};
            \def\radius{8pt}
            \draw [fill=gray](zi) circle [radius=\radius];
            \draw [fill=gray](zim1) circle [radius=\radius];
            \draw [fill=gray](zim2) circle [radius=\radius];
            \draw [fill=gray](zip1) circle [radius=\radius];
            \draw [fill=gray](zip2) circle [radius=\radius];
            \draw [fill=gray](zip3) circle [radius=\radius];
            \draw (zi) node[below=\radius/2] {$z_{i}$};
            \draw (zim1) node[below=\radius/2] {$z_{i-1}$};
            \draw (zip1) node[below left=\radius/2] {$z_{i+1}$};
            \def\dotradius{2pt}
            \draw[] let \p1=(zim2), \p2=(zip3) in
                [fill=black](\p2)+(\x1/2-\x2/2,\y1/2-\y2/2) circle [radius=\dotradius];
            \draw[] let \p1=(zim2), \p2=(zip3) in
                [fill=black](\p2)+(\x1/3-\x2/3,\y1/3-\y2/3) circle [radius=\dotradius];
            \draw[] let \p1=(zim2), \p2=(zip3) in
                [fill=black](\p2)+(\x1*2/3-\x2*2/3,\y1*2/3-\y2*2/3) circle [radius=\dotradius];

\end{tikzpicture}
  \caption{An illustration of circular formations.}
  \label{fig_circular_formation}
\end{figure}

Consider $n$ ($n\ge3$) vehicles in the plane.
The target formation considered in this paper is an angle-constrained circular (or polygon) formation.
The underlying information flow among the vehicles is described by an undirected circular graph with fixed topology.
By indexing the vehicles properly, we can have $\mathcal{N}_i=\{i-1,i+1\}$ for $i\in\{1,\dots,n\}$,
which means vehicle $i$ can measure the bearings of vehicles $i-1$ and $i+1$.
Note that the indices $i-1$ and $i+1$ are taken modulo $n$ in this paper.
Denote the position of vehicle $i$ as $z_i\in\mathbb{R}^2$, and the edge between vehicles $i$ and $i+1$ as $e_i=z_{i+1}-z_i$.
The unit-length vector $g_i={e_i}/{\|e_i\|}$ characterizes the relative bearing between vehicles $i$ and $i+1$ (see Figure~\ref{fig_circular_formation}).
Hence the measurements of vehicle $i$ consist of $g_{i}$ and $-g_{i-1}$.
It should be noted that vehicle $i$ may measure the bearings $g_{i}$ and $-g_{i-1}$ in its local coordinate frame.
But in order to analyze the dynamics of the entire system, we need to write these bearing measurements in a global coordinate frame.

The angle subtended by vehicles $i-1$ and $i+1$ at vehicle $i$ is denoted as $\theta_i\in[0,2\pi)$.
More specifically, rotating $-g_{i-1}$ counterclockwise through an angle $\theta_i$ about vehicle $i$ yields $g_i$ (see Figure~\ref{fig_circular_formation}).
That can be mathematically expressed as
\begin{align}\label{eq_gigi-1}
    g_i=R(\theta_i)(-g_{i-1}),
\end{align}
where $R(\cdot)$ is the rotation matrix in \eqref{eq_rotation_matrix}.
By defining $\theta_i$ in the above way, $\theta_i$ and $\theta_{i+1}$ are on the same side of edge $e_i$ for all $i\in\{1,\dots,n\}$.
Consequently the quantity $\sum_{i=1}^n \theta_i$ is invariant to the positions of the vehicles because the sum of the interior angles of a polygon is constant.
In the target formation, the angle $\theta_i$ is specified as $\theta_i^*\in[0,2\pi)$.
Hence if $\sum_{i=1}^n \theta_i(0)=\sum_{i=1}^n \theta_i^*$, then we have $\sum_{i=1}^n \theta_i\equiv\sum_{i=1}^n \theta_i^*$.

The target angles $\{\theta_i^*\}_{i=1}^n$ should be feasible such that there exist $\{z_i\}_{i=1}^n$ ($z_i\ne z_j$ for $i\ne j$) to realize the target formation.
Since the target formation is constrained only by angles, its realization would be non-unique.
Specifically, the orientation, translation and scale of the target formation is non-unique.
Moreover, since we make no assumptions about parallel rigidity \cite{eren2003,bishopconf2011rigid,ErenIJC}, the shape of the target formation may not be unique either.
In fact, the shape of a circular formation cannot be uniquely determined by specifying the angles unless $n=3$.
In order to control the shape of the formation using bearing-only measurements, the underlying information flow should be more complicated than a circular graph.
For example, each vehicles should correspond to more than one angle.
We leave formation shape control using bearing-only measurements for future work.

\subsection{Proposed Control Law}
Suppose that no vehicles are collocated in the initial formation, i.e., $z_i(0)\ne z_j(0)$ for all $i\ne j$.
Consider the dynamics of each vehicle as a single integrator: $\dot{z}_i=u_i$.
Our task is to design $u_i$ to steer vehicles from their initial positions to a target formation.
The angle error corresponding to vehicle $i$ is chosen as
\begin{align}\label{eq_error_definition}
    \varepsilon_i
    =\cos\theta_i-\cos\theta_i^* =-g_i^{\T}g_{i-1}-\cos\theta_i^*.
\end{align}
The reason why we use a cosine function to define the angle error $\varepsilon_i$ is that $\cos\theta_i$ can be conveniently expressed as the inner product of the two bearing measurements $g_i$ and $-g_{i-1}$.
The proposed control law for vehicle $i$ is
\begin{align}\label{eq_controlLaw}
        \dot{z}_i=\sgn(\varepsilon_i)(g_i-g_{i-1}),
\end{align}
where
\begin{align*}
        \sgn(\varepsilon_i)=\left\{
          \begin{array}{l l}
            1  & \mbox{if } \varepsilon_i>0\\
            0  & \mbox{if } \varepsilon_i=0\\
            -1 & \mbox{if } \varepsilon_i<0
          \end{array}.\right.
\end{align*}
For vector arguments, $\sgn(\cdot)$ is defined component-wise.

\begin{remark}
    Compared to the control laws in \cite{bishopSCL,zhaoArxiv}, the one \eqref{eq_controlLaw} also steers the vehicles along the bisectors of their corresponding angles, respectively.
    But the control law \eqref{eq_controlLaw} is discontinuous and only requires the sign information of the angle errors.
    Due to the discontinuity, classical Lyapunov approaches are inapplicable here.
    We will prove the stability of the closed-loop system by using a locally Lipschitz Lyapunov function and nonsmooth analysis tools.
\end{remark}

From Figure \ref{fig_circular_formation} or equation \eqref{eq_gigi-1}, it is obvious to see that $g_i-g_{i-1}=0$ when $\theta_i=\pi$.
Hence the control law \eqref{eq_controlLaw} would be ineffective in the case of $\theta_i=\pi$ even though $\varepsilon_i$ is still nonzero.
Moreover, when $\theta_i=0$, vehicles $i-1$ and $i+1$ are located on the same side of vehicle $i$.
Since bearing is usually measured by optical sensors such as cameras, the bearing of vehicle $i-1$ or $i+1$ may not measurable by vehicle $i$ due to line-of-sight occlusion in the case of $\theta_i=0$.
Therefore, we adopt the following assumption.
\begin{assumption}\label{assumption}
    In the target formation, $\theta_i^*\ne 0$ and $\theta_i^*\ne \pi$ for all $i\in\{1,\dots,n\}$.
\end{assumption}
By Assumption \ref{assumption}, the target angle $\theta_i^*$ is in either $(0,\pi)$ or $(\pi,2\pi)$.
In other words, no three consecutive vehicles in the target formation are collinear.
The collinear case is a difficulty in many formation control problems (see, for example, \cite{Francis2009IJC,Francis2010TAC,Huang2012,bishopconf2011quad}).
At last, in order to analyze the dynamics of the whole system, we need to write the bearings $g_i$ for all $i\in\{1,\dots,n\}$ in a global coordinate frame.
But the control law \eqref{eq_controlLaw} can be implemented distributedly even if $g_i$ and $-g_{i-1}$ are measured in the local coordinate frame of vehicle $i$.

\subsection{Error Dynamics}
Denote $\varepsilon=[\varepsilon_1,...,\varepsilon_n]^{\T}\in\mathbb{R}^n$.
Now we derive the dynamics of $\varepsilon$.
Since $g_i=e_i/\|e_i\|$, the time derivative of $g_i$ is
\begin{align}\label{eq_gi_dot}
    \dot{g}_i=\frac{1}{\|e_i\|}P_i\dot{e}_i,
\end{align}
where $P_i=I-g_ig_i^{\T}$. Note $P_i$ is an orthogonal projection matrix satisfying $P_i^{\T}=P_i$ and $P_i^2=P_i$.
Moreover, $P_i$ is positive semi-definite and $\Null(P_i)=\mathrm{span}\{g_i\}$.
Since $e_i=z_{i+1}-z_i$, by the control law \eqref{eq_controlLaw}, the time derivative of $e_i$ is given by
\begin{align}\label{eq_e_dot}
    \dot{e}_i \nonumber
    &=\dot{z}_{i+1}-\dot{z}_{i} \\ \nonumber
    &=\sgn(\varepsilon_{i+1})(g_{i+1}-g_{i}) - \sgn(\varepsilon_i)(g_i-g_{i-1}) \\
    &=\sgn(\varepsilon_{i+1})g_{i+1}+\sgn(\varepsilon_i)g_{i-1} - \left[\sgn(\varepsilon_{i+1})+\sgn(\varepsilon_{i})\right] g_i.
\end{align}
Substituting \eqref{eq_e_dot} into \eqref{eq_gi_dot} and using the fact that $P_ig_i=0$ yield
\begin{align*}
    \dot{g}_i=\frac{1}{\|e_i\|}P_i\left[ \sgn(\varepsilon_{i+1})g_{i+1}+\sgn(\varepsilon_i)g_{i-1} \right].
\end{align*}
Recall $\varepsilon_i=-g_i^{\T}g_{i-1}-\cos(\theta_i^*)$ as shown in \eqref{eq_error_definition} and $\theta_i^*$ is constant.
Then
\begin{align*}
    \dot{\varepsilon}_i \nonumber
    &=-g_i^{\T}\dot{g}_{i-1}-g_{i-1}^{\T}\dot{g}_i \\ \nonumber
    &=-\frac{1}{\|e_{i-1}\|}g_i^{\T}P_{i-1}\left[ \sgn(\varepsilon_{i})g_{i}+\sgn(\varepsilon_{i-1})g_{i-2} \right]
    -\frac{1}{\|e_i\|}g_{i-1}^{\T}P_i\left[ \sgn(\varepsilon_{i+1})g_{i+1}+\sgn(\varepsilon_i)g_{i-1} \right] \\ \nonumber
    &= - a_{i(i-1)}\sgn(\varepsilon_{i-1}) - a_{ii}\sgn(\varepsilon_i) - a_{i(i+1)}\sgn(\varepsilon_{i+1}),
\end{align*}
where
\begin{align*}
    a_{i(i-1)} &= \frac{1}{\|e_{i-1}\|}g_i^{\T}P_{i-1}g_{i-2}, \nonumber \\
    a_{ii}     &= \frac{1}{\|e_{i-1}\|}g_i^{\T}P_{i-1}g_{i} + \frac{1}{\|e_i\|}g_{i-1}^{\T}P_ig_{i-1}, \nonumber \\
    a_{i(i+1)} &= \frac{1}{\|e_i\|}g_{i-1}^{\T}P_ig_{i+1}.
\end{align*}
Hence the dynamics of $\varepsilon$ can be written as
\begin{align}\label{eq_errorSystem}
    \dot{\varepsilon}=-A\sgn(\varepsilon),
\end{align}
where $[A]_{i(i-1)}=a_{i(i-1)}$, $[A]_{ii}=a_{ii}$ and $[A]_{i(i+1)}=a_{i(i+1)}$ for all $i\in\{1,\dots,n\}$; and all the other entries of $A$ are zero.
By changing the index $i$ of $a_{i(i+1)}$ to $i-1$, we can obtain the formula of $a_{(i-1)i}$.
It is easy to see that $a_{(i-1)i}=a_{i(i-1)}$ for all $i$ and hence $A$ is symmetric.
The next lemma shows that $A$ is also positive semi-definite.

\begin{lemma}\label{lemma_A_positivedefinite}
    For any $x=[x_1,\dots,x_n]^\T\in\mathbb{R}^n$,
    \begin{align}\label{eq_A_is_positive_definite}
        x^\T Ax
        =\sum_{i=1}^n \frac{1}{\|e_i\|}\left(g_{i+1}x_{i+1}+g_{i-1}x_i\right)^{\T} P_i \left(g_{i+1}x_{i+1}+g_{i-1}x_i\right)
        \ge 0.
    \end{align}
    As a result, the matrix $A$ in \eqref{eq_errorSystem} is positive semi-definite.
\end{lemma}
\begin{proof}
For any vector $x=[x_1,\dots,x_n]^\T\in\mathbb{R}^n$, we have
    \begin{align*}
        x^{\T} A x
        &=\sum_{i=1}^n a_{i(i-1)}x_ix_{i-1} + a_{ii}x_i^2 + a_{i(i+1)}x_ix_{i+1} \nonumber \\
        &= \sum_{i=1}^n \left(\frac{1}{\|e_{i-1}\|}g_i^{\T}P_{i-1}g_{i-2}\right) x_ix_{i-1}
         +\sum_{i=1}^n\left( \frac{1}{\|e_{i-1}\|}g_i^{\T}P_{i-1}g_{i}\right)x_{i}^2 \nonumber \\
         &\qquad \quad +\sum_{i=1}^n \left(\frac{1}{\|e_i\|}g_{i-1}^{\T}P_ig_{i-1} \right)x_i^2
         +\sum_{i=1}^n \left(\frac{1}{\|e_i\|}g_{i-1}^{\T}P_ig_{i+1}\right) x_ix_{i+1} \nonumber \\
        &= \sum_{i=1}^n \left(\frac{1}{\|e_{i}\|}g_{i+1}^{\T}P_{i}g_{i-1}\right) x_{i+1}x_{i}
         +\sum_{i=1}^n\left( \frac{1}{\|e_{i}\|}g_{i+1}^{\T}P_{i}g_{i+1}\right)x_{i+1}^2 \nonumber \\
         &\qquad \quad +\sum_{i=1}^n \left(\frac{1}{\|e_i\|}g_{i-1}^{\T}P_ig_{i-1} \right)x_i^2
         +\sum_{i=1}^n \left(\frac{1}{\|e_i\|}g_{i-1}^{\T}P_ig_{i+1}\right) x_ix_{i+1} \nonumber \\
        &=\sum_{i=1}^n \frac{1}{\|e_i\|}\left(g_{i+1}x_{i+1}+g_{i-1}x_i\right)^{\T} P_i \left(g_{i+1}x_{i+1}+g_{i-1}x_i\right)
        \ge 0,
    \end{align*}
    where the last inequality uses the fact that $P_i$ is positive semi-definite.
\end{proof}

\section{Analysis of Formation Stability and Behavior}\label{section_stabilityAnalysis}
The stability of the error dynamics \eqref{eq_errorSystem} is analyzed in this section.
By employing a locally Lipschitz Lyapunov function and the nonsmooth analysis tools introduced in Section \ref{subsection_nonsmoothanalysis}, we prove that the origin $\varepsilon=0$ is locally finite-time stable with collision avoidance guaranteed.
In addition to the dynamics of $\varepsilon$, we also analyze the behaviors of the vehicle positions during formation convergence.

\subsection{Nonsmooth Lyapunov Function}
Consider the Lyapunov function
\begin{align*}
    V(\varepsilon)=\sum_{i=1}^n |\varepsilon_i|,
\end{align*}
which is positive definite with respect to $\varepsilon$.
Note $V(\varepsilon)$ is locally Lipschitz and convex. Hence $V(\varepsilon)$ is also regular.
By the definition of the generalized gradient, we have
\begin{align*}
    \partial V(\varepsilon)=\{ \eta=[\eta_1,\dots,\eta_n]^{\T}\in\mathbb{R}^n \ |\
        & \eta_i=\sgn(\varepsilon_i) \mbox{ if } \varepsilon_i\ne 0 \mbox{ and } \\
        & \eta_i\in[-1,1] \mbox{ if } \varepsilon_i=0 \mbox{ for } i\in\{1,\dots,n\}\}.
\end{align*}
Because $|\eta_i|=|\sgn(\varepsilon_i)|=1$ if $\varepsilon_i\ne0$, we have the obvious but important fact that
\begin{align}\label{eq_eta_norm}
    \|\eta\|\ge 1, \quad \forall \eta\in\partial V(\varepsilon), \ \forall \varepsilon\ne0.
\end{align}
In addition, if $\varepsilon_i\ne0$, $\Ln(\{\sgn(\varepsilon_i)\})=\{\sgn(\varepsilon_i)\}$; and if $\varepsilon_i=0$, $\Ln([-1,1])=\{0\}=\{\sgn(0)\}$. Thus we have the following useful property
\begin{align}\label{eq_LnpartialV}
    \Ln(\partial V(\varepsilon))=\{\sgn(\varepsilon)\}.
\end{align}

\subsection{Calculate the Filippov Differential Inclusion}
Consider the error dynamics in \eqref{eq_errorSystem}.
First, the term $\sgn(\varepsilon)$ in \eqref{eq_errorSystem} is discontinuous in $\varepsilon$.
Second, it is noticed that $\|e_i\|$ appears in the denominators of the nonzero entries of $A$.
Hence if $\|e_i\|$ can be zero, the term $A$ is also discontinuous.
Note that $\|e_i\|$ being zero simply means that the vehicles $i$ and $i+1$ are colliding with each other.
In the initial formation, it is assumed that no vehicles are collocated, i.e., $z_i(0)\ne z_j(0)$ for any $i\ne j$.
By the control law \eqref{eq_controlLaw} we have $\|\dot{z}_i\|\le\|g_i-g_{i-1}\|\le2$, which means that the maximum speed of each vehicle is $2$.
Thus
\begin{align}\label{eq_collision_time}
    T^*=\frac{\min_{i\ne j}\|z_i(0)-z_j(0)\|}{4}
\end{align}
is the minimum time when any two vehicles could possibly collide with each other.
In other words, when $t<T^*$, no vehicles collide with each other, i.e., $\|e_i(t)\|\ne0$.
In the rest of the paper, we will only consider $t\in[0,T]$ with $T<T^*$.
We will prove that the system can be stabilized within the finite time interval $[0,T]$.

Since $\|e_i(t)\|\ne0$ for all $i$ and all $t\in[0,T]$, the matrix $A$ is continuous.
Then by \cite[Theorem 1, 5)]{1987Paden}, the Filippov differential inclusion associated with the system \eqref{eq_errorSystem} can be calculated as
\begin{align}\label{eq_myFilippovInclusion}
    \dot{\varepsilon}
    &\in \mathcal{F}[-A\sgn(\varepsilon)] = -A\mathcal{F}[\sgn(\varepsilon)].
\end{align}
Because $\{\sgn(\varepsilon)\} = \Ln(\partial V(\varepsilon))$ as given in \eqref{eq_LnpartialV}, we have
\begin{align*}
    \mathcal{F}[\sgn(\varepsilon)]=\mathcal{F}[\Ln(\partial V(\varepsilon))]=\partial V(\varepsilon),
\end{align*}
where the last equality uses the fact \eqref{eq_FilippovFact}.
Thus the Filippov differential inclusion in \eqref{eq_myFilippovInclusion} can be rewritten as
\begin{align}\label{eq_differential_inclusion}
    \dot{\varepsilon}\in -A\partial V(\varepsilon).
\end{align}

\subsection{Calculate the Set-valued Lie Derivative}
The set-valued Lie derivative of $V(\varepsilon)$ with respect to \eqref{eq_differential_inclusion} is given by
\begin{align}\label{eq_set_value_Lie}
    \widetilde{\mathcal{L}}_{-A\partial V}V(\varepsilon)
    &=\{\ell\in\mathbb{R} \ |\ \exists \xi\in -A \partial V(\varepsilon),\ \forall \zeta\in\partial V(\varepsilon),\  \zeta^{\T} \xi=\ell\} \nonumber \\
    &=\{\ell\in\mathbb{R} \ |\ \exists \eta\in \partial V(\varepsilon),\  \forall \zeta\in\partial V(\varepsilon),\  -\zeta^{\T} A \eta=\ell\}.
\end{align}
When $\widetilde{\mathcal{L}}_{-A\partial V}V(\varepsilon)\ne \emptyset$, for any $\ell\in\widetilde{\mathcal{L}}_{-A\partial V}V(\varepsilon)$, there exists $\eta\in\partial V$ such that $\ell=-\zeta^{\T} A \eta$ for all $\zeta\in\partial V$.
In particular, by choosing $\zeta=\eta$ we have
\begin{align}\label{eq_a}
    \ell=-\eta^{\T} A \eta \le 0.
\end{align}
The last inequality is because $A$ is a positive semi-definite matrix as shown in Lemma \ref{lemma_A_positivedefinite}.
Now we have $\widetilde{\mathcal{L}}_{-A\partial V}V(\varepsilon)=\emptyset$ or $\max \widetilde{\mathcal{L}}_{-A\partial V}V(\varepsilon)\le0$.

\subsection{Main Convergence Result}
We need to introduce the following results before presenting our main convergence result.

Given an angle $\alpha\in\mathbb{R}$ and a vector $x\in\mathbb{R}^2$, the angle between $x$ and $R(\alpha)x$ is $\alpha$.
Thus for all nonzero $x\in\mathbb{R}^2$, $x^{\T}R(\alpha)x>0$ when $\alpha\in(-\pi/2,\pi/2)$ (mod $2\pi$); $x^{\T}R(\alpha)x=0$ when $\alpha=\pm\pi/2$ (mod $2\pi$); and $x^{\T}R(\alpha)x<0$ when $\alpha\in(\pi/2,3\pi/2)$ (mod $2\pi$).

\begin{lemma}\label{lemma_gi_perp}
Let $g_i^{\perp}=R(\pi/2)g_i$.
It is obvious that $\|g_i^{\perp}\|=1$ and $(g_i^{\perp})^{\T}g_i=0$.
Furthermore,
\begin{enumerate} [(i)]
  \item $ P_i=g_i^{\perp}(g_i^{\perp})^{\T}$.
  \item For $i\ne j$, $(g_i^{\perp})^{\T}g_j=-(g_j^{\perp})^{\T} g_i$.
  \item $(g_i^{\perp})^{\T}g_{i-1}=\sin\theta_i$. Consequently, $(g_i^{\perp})^{\T}g_{i-1}>0$ if $\theta_i\in(0,\pi)$; and $(g_i^{\perp})^{\T}g_{i-1}<0$ if $\theta_i\in(\pi,2\pi)$.
\end{enumerate}
\end{lemma}
\begin{proof}
    See \cite[Lemma 5]{zhaoArxiv}.
\end{proof}
\begin{lemma}\label{lemma_infimum_angle}
    Let $B\in\mathbb{R}^{n\times n}$ be a positive semi-definite matrix with $\lambda_1(B)=0$ and $\lambda_2(B)>0$.
    An eigenvector associated with the zero eigenvalue is $\one=[1,\dots,1]^{\T}\in\mathbb{R}^n$.
    Let
    \begin{align*}
        \mbox{
        $\mathcal{U}=\{x\in\mathbb{R}^n \ | \  \|x\|=1$ and nonzero entries of $x$ are not with the same sign\}.
        }
    \end{align*}
    Then
    \begin{align*}
        \inf_{x\in\mathcal{U}} x^{\T}Bx = \frac{\lambda_2(B)}{n}.
    \end{align*}
\end{lemma}
\begin{proof}
    See \cite[Lemma 1]{zhaoArxiv}.
\end{proof}

Now we are ready to examine the elements in $\widetilde{\mathcal{L}}_{-A\partial V}V(\varepsilon)$ more closely.
Note if $\widetilde{\mathcal{L}}_{-A\partial V}V(\varepsilon)=\emptyset$, we have $\max \widetilde{\mathcal{L}}_{-A\partial V}V(\varepsilon)=-\infty$.
Hence we need only to focus on the case that $\widetilde{\mathcal{L}}_{-A\partial V}V(\varepsilon)\ne \emptyset$.
Recall for any $\ell \in \widetilde{\mathcal{L}}_{-A\partial V}V(\varepsilon)$, there exists $\eta\in\partial V$ such that $\ell=-\eta^TA\eta$ as shown in \eqref{eq_a}.
By \eqref{eq_A_is_positive_definite}, we can further write $\ell=-\eta^TA\eta$ as
\begin{align}\label{eq_etaAeta}
        \ell
        &=-\sum_{i=1}^n \frac{1}{\|e_i\|} \left(g_{i+1}\eta_{i+1}+g_{i-1}\eta_i\right)^{\T} P_i \left(g_{i+1}\eta_{i+1}+g_{i-1}\eta_i\right) \nonumber \\
        &\le -\frac{1}{\sum_{i=1}^n \|e_i\|}\sum_{i=1}^n \left(g_{i+1}\eta_{i+1}+g_{i-1}\eta_i\right)^{\T} P_i \left(g_{i+1}\eta_{i+1}+g_{i-1}\eta_i\right) \nonumber \\
        &= -\frac{1}{\sum_{i=1}^n \|e_i\|}\sum_{i=1}^n \left[\left(g_{i+1}\eta_{i+1}+g_{i-1}\eta_i\right)^{\T} g_i^\perp\right]^2 \quad \mbox{(by Lemma \ref{lemma_gi_perp} (i) )} \nonumber \\
        &= -\frac{1}{\sum_{i=1}^n \|e_i\|}\sum_{i=1}^n \left[(g_i^\perp)^{\T}g_{i+1}\eta_{i+1}+(g_i^\perp)^{\T}g_{i-1}\eta_i\right]^2 \nonumber \\
        &= - \frac{1}{\sum_{i=1}^n \|e_i\|} h^{\T}h,
\end{align}
    where
\begin{align}\label{eq_w}
       h
    &=\left[
      \begin{array}{c}
        (g_1^{\perp})^{\T}g_{2} \eta_{2} +(g_1^{\perp})^{\T}g_{n} \eta_1\\
        \vdots \\
        (g_n^{\perp})^{\T}g_{1} \eta_{1} +(g_n^{\perp})^{\T}g_{n-1} \eta_n \\
      \end{array}
    \right] \nonumber\\
    &=\left[
          \begin{array}{ccccc}
            (g_1^{\perp})^{\T}g_n & (g_1^{\perp})^{\T}g_2 & 0 & \dots & 0 \\
            0                  & (g_2^{\perp})^{\T}g_1 & (g_2^{\perp})^{\T}g_3 & \dots & 0 \\
            0                  & 0                  & (g_3^{\perp})^{\T}g_2 & \dots & 0 \\
            \vdots             & \vdots & \vdots & \ddots & \vdots \\
            (g_n^{\perp})^{\T}g_1 & 0 & \dots & 0 & (g_n^{\perp})^{\T}g_{n-1} \\
          \end{array}
        \right]
        \left[
          \begin{array}{c}
            \eta_1 \\
            \eta_2 \\
            \eta_3 \\
            \vdots \\
            \eta_n \\
          \end{array}
        \right] \nonumber\\
    &=ED\eta
\end{align}
with
\begin{align*}
    E = \left[
          \begin{array}{ccccc}
            1      & -1     & 0      & \dots & 0 \\
            0      & 1      & -1     & \dots & 0 \\
            0      & 0      & 1      & \dots & 0 \\
            \vdots & \vdots & \vdots & \ddots & \vdots \\
            -1     & 0      & \dots  & 0 & 1 \\
          \end{array}
        \right]\in\mathbb{R}^{n\times n}, \,\,
    D =   \left[
          \begin{array}{ccccc}
            (g_1^{\perp})^{\T}g_n & 0 & 0 & \dots & 0 \\
            0                  & (g_2^{\perp})^{\T}g_1 & 0 & \dots & 0 \\
            0                  & 0 & (g_3^{\perp})^{\T}g_2 & \dots & 0 \\
            \vdots & \vdots & \vdots & \ddots & \vdots \\
            0 & 0 & \dots & 0 & (g_n^{\perp})^{\T}g_{n-1} \\
          \end{array}
        \right]\in\mathbb{R}^{n\times n}.
\end{align*}
The last equality of \eqref{eq_w} uses the fact that $(g_i^{\perp})^{\T}g_{i-1}=-(g_{i-1}^{\perp})^{\T}g_{i}$ as shown in Lemma~\ref{lemma_gi_perp}~(ii).
Note that $D$ is a diagonal matrix and $E$ actually is an incidence matrix of a directed and connected circular graph.
Substituting \eqref{eq_w} into \eqref{eq_etaAeta} gives
\begin{align}\label{eq_a1}
        \ell=-\eta^{\T} A \eta \le -\frac{1}{\sum_{i=1}^n \|e_i\|} \eta^{\T} D^{\T} E^{\T} E D \eta.
\end{align}

We now present the main stability result.

\begin{theorem}\label{theorem_stability}
    Under Assumption \ref{assumption}, if no vehicles are collocated in the initial formation, i.e., $z_i(0)\ne z_j(0)$ for $i\ne j$ and $i,j\in\{1,\dots,n\}$, the equilibrium $\varepsilon=0$ of system \eqref{eq_errorSystem} is locally finite-time stable.
\end{theorem}
\begin{proof}
Consider the time interval $[0,T]$ with $T<T^*$.
The minimum collision time $T^*$ is given in \eqref{eq_collision_time}.
Hence for all $t\in[0,T]$, we have $\|e_i(t)\|\ne0$ and $\|e_i(t)\|\ne+\infty$.
We will prove that $\varepsilon$ can converge to zero in the finite time interval $[0,T]$ if $\varepsilon(0)$ is sufficiently small.

Let $\Omega(\varepsilon(0))=\{\varepsilon\in\mathbb{R}^n \ |\ V(\varepsilon)\le V(\varepsilon(0))\}$.
Since $V(\varepsilon)=\sum_{i=1}^n |\varepsilon_i|=\|\varepsilon\|_1$, the level set $\Omega(\varepsilon(0))$ is connected and compact.
Because $\widetilde{\mathcal{L}}_{-A\partial V}V(\varepsilon)=\emptyset$ or $\max \widetilde{\mathcal{L}}_{-A\partial V}V(\varepsilon)\le 0$ for any $\varepsilon\in\Omega(\varepsilon(0))$, we know that $\Omega(\varepsilon(0))$ is strongly invariant for \eqref{eq_errorSystem} over $[0,T]$ by Lemma \ref{lemma_invariancePrinciple}.

Denote $\delta_i=\theta_i-\theta_i^*$ and $\delta=[\delta_1,\dots,\delta_n]^{\T}\in\mathbb{R}^n$.
Because $\sum_{i=1}^n \theta_i \equiv \sum_{i=1}^n \theta_i^*$, we have $\sum_{i=1}^n \delta_i = 0$.
Thus if $\delta\ne0$, the nonzero entries of $\delta$ are \emph{not} with the same sign.
Let
\begin{align*}
    w_i=\frac{\cos \theta_i - \cos \theta_i^*}{\theta_i-\theta_i^*}.
\end{align*}
Then $\varepsilon_i=w_i\delta_i$ and hence
\begin{align*}
    \varepsilon=W\delta,
\end{align*}
where $W=\mathrm{diag}\{w_1,\dots,w_n\}\in\mathbb{R}^{n\times n}$.
Since $\lim_{\theta_i\rightarrow\theta_i^*} w_i=-\sin\theta_i^*$, the equations $\varepsilon_i=w_i\delta_i$ and $\varepsilon = W \delta$ are always valid even when $\theta_i-\theta_i^*=0$.
There exists sufficiently small $V(\varepsilon(0))$ such that $\theta_i(0)$ is sufficiently close to $\theta_i^*$ and hence $\theta_i,\theta_i^*\in(0,\pi)$ or $\theta_i,\theta_i^*\in(\pi,2\pi)$ for all $\varepsilon\in\Omega(\varepsilon(0))$.
Note that $w_i<0$ if $\theta_i,\theta_i^*\in(0,\pi)$, and $w_i>0$ if $\theta_i,\theta_i^*\in(\pi,2\pi)$.
Moreover, recall $(g_i^{\perp})^{\T}g_{i-1}>0$ when $\theta_i\in(0,\pi)$, and $(g_i^{\perp})^{\T}g_{i-1}<0$ when $\theta_i\in(\pi,2\pi)$ as shown in Lemma \ref{lemma_gi_perp} (iii).
Thus we have
\begin{align*}
    (g_i^{\perp})^{\T}g_{i-1}w_i<0
\end{align*}
for all $i\in\{1,\dots,n\}$ and consequently the diagonal entries of $DW$ are with the same sign.
Suppose $\varepsilon\ne0$ and hence $\delta\ne0$.
Because the nonzero entries in $\delta$ are \emph{not} with the same sign, the nonzero entries of $DW\delta$ and hence $D\varepsilon$ are \emph{not} with the same sign either.
Furthermore, because $\eta_i=\sgn(\varepsilon_i)$ if $\varepsilon_i\ne 0$, the nonzero entry $\varepsilon_i$ has the same sign with $\eta_i$.
As a result, the nonzero entries of $D\eta$ are \emph{not} with the same sign either.
Thus we have
\begin{align*}
    \frac{D\eta}{\|D\eta\|}\in \mathcal{U}
\end{align*}
with $\mathcal{U}$ defined in Lemma \ref{lemma_infimum_angle}.
In addition, note $E$ is an incidence matrix of a directed and connected circular graph.
By \cite[Theorem 8.3.1]{graphbook}, we have $\rank(E)=n-1$ and $\Null(E^{\T}E)=\Null(E)=\{\one\}$.
Thus inequality \eqref{eq_a1} implies
\begin{align}\label{eq_etaAeta_upperbound}
    \ell
    &=-\eta^{\T} A \eta \nonumber\\
    &\le -\frac{1}{\sum_{i=1}^n \|e_i\|} \frac{\lambda_2(E^{\T}E)}{n}\|D\eta\|^2 \quad \mbox{(by Lemma \ref{lemma_infimum_angle})}\nonumber \\
    &\le -\frac{1}{\sum_{i=1}^n \|e_i\|} \frac{\lambda_2(E^{\T}E)}{n}\lambda_1(D^2)\|\eta\|^2 \nonumber \\
    &\le -\frac{1}{\sum_{i=1}^n \|e_i\|} \frac{\lambda_2(E^{\T}E)}{n}\lambda_1(D^2),
\end{align}
where the last inequality uses the fact $\|\eta\|\ge1$ if $\varepsilon\ne0$ as shown in \eqref{eq_eta_norm}.

Now we examine the terms $\sum_{i=1}^n \|e_i\|$ and $\lambda_1(D^2)$ in \eqref{eq_etaAeta_upperbound}.
First, over the finite time interval $[0,T]$, the quantity $\sum_{i=1}^n \|e_i\|$ cannot go to infinity because the vehicle speed is finite.
Hence there exists a constant $\gamma>0$ such that $\sum_{i=1}^n \|e_i\|\le \gamma$.
Second, since $D$ is diagonal, we have $\lambda_1(D^2)=\min_{i}[D]_{ii}^2$.
At the equilibrium point $\varepsilon=0$ (i.e., $\theta_i=\theta_i^*$ for all $i$), we have $[D]_{ii}=(g_i^{\perp})^{\T}g_{i-1}\ne 0$ because $\theta_i^*\ne 0$ or $\pi$ as stated in Assumption \ref{assumption}.
By continuity, we can still have $[D]_{ii}\ne 0$ for all $\varepsilon\in \Omega(\varepsilon(0))$ if $\varepsilon(0)$ is sufficiently small.
Because $\Omega(\varepsilon(0))$ is compact, there exist a lower bound $\beta$ such that $\lambda_1(D^2)\ge \beta$ for all $\varepsilon\in\Omega(\varepsilon(0))$.
Then \eqref{eq_etaAeta_upperbound} can be rewritten as
\begin{align}\label{eq_etaAeta_upperbound_final}
    \ell
    &=-\eta^{\T} A \eta \le -\frac{\beta\lambda_2(E^{\T}E)}{\gamma n} \triangleq -\kappa<0, \quad \forall \varepsilon\in\Omega(\varepsilon(0))\setminus\{0\}.
\end{align}

If $\varepsilon=0$ we have $0\in\widetilde{\mathcal{L}}_{-A\partial V}V(\varepsilon)$ because of \eqref{eq_set_value_Lie} and the fact that $0\in\partial V(0)$; and if $\varepsilon\ne0$ we have $0\notin\widetilde{\mathcal{L}}_{-A\partial V}V(\varepsilon)$ because $\max \widetilde{\mathcal{L}}_{-A\partial V}V(\varepsilon)<0$ by \eqref{eq_etaAeta_upperbound_final}.
Thus by the definition \eqref{eq_Z_f_V}, we have
\begin{align}\label{eq_equilibrium}
    Z_{-A\sgn(\varepsilon),V(\varepsilon)}=\{0\}.
\end{align}
Based on \eqref{eq_etaAeta_upperbound_final}, \eqref{eq_equilibrium} and Lemma \ref{lamma_finite_time}, any solution of \eqref{eq_errorSystem} starting from $\varepsilon(0)$ converges to $\varepsilon=0$ in finite-time.
The convergence time is upper bounded by $V(\varepsilon(0))/\kappa$.
If $V(\varepsilon(0))$ is sufficiently small, we can have
\begin{align*}
    \frac{V(\varepsilon(0))}{\kappa}<T<T^*,
\end{align*}
which means that the system can be stabilized within the time interval $[0,T]$.
\end{proof}

\subsection{Formation Behavior}\label{section_formationBehavior}
Because the target formation is constrained only by angles, the positions of the vehicles or the inter-vehicle distance are not specified in the final converged formation.
In addition to the dynamics of $\varepsilon$, it is also important to study the evolution of the vehicle positions $z=\left[z_1^{\T},\dots,z_n^{\T}\right]^{\T}\in\mathbb{R}^{2n}$.
Next we identify a number of behaviors of the formation controlled by the control law \eqref{eq_controlLaw}.

Firstly, from the control law \eqref{eq_controlLaw}, it is trivial to see that $\dot{z}=0$ if $\varepsilon=0$, which means that all vehicles will stop moving if all angle errors have converged to zero.

Secondly, recall the error dynamics is given by $\dot{\varepsilon}=-A\sgn(\varepsilon)$ as shown in \eqref{eq_errorSystem}.
Similar to the derivation of \eqref{eq_a1}, it can be shown that
\begin{align*}
    \sgn(\varepsilon)^{\T}A\sgn(\varepsilon)\ge \frac{1}{\sum_{i=1}^n\|e_i\|}\sgn(\varepsilon)^{\T}D^{\T}E^{\T}ED\sgn(\varepsilon).
\end{align*}
Furthermore, analogous to \eqref{eq_etaAeta_upperbound}, we have $\sgn(\varepsilon)^{\T}A\sgn(\varepsilon)=0$ if and only if $\varepsilon=0$
though $A$ is merely positive semi-definite.
Since $A\sgn(\varepsilon)=0$ if and only if $\sgn(\varepsilon)^{\T}A\sgn(\varepsilon)=0$, we obtain that $\dot{\varepsilon}=-A\sgn(\varepsilon)=0$ if and only if $\varepsilon=0$.
As a result, as long as the angle errors are nonzero, the angles will keep changing.
Hence it is impossible that the formation is moving while all angles are not changing.
In other words, we can rule out the possibility that only the orientation, translation or scale of the formation is changing while the angles are not.

Thirdly, suppose the target formation is achieved at time $t_f$.
The proof of Theorem \ref{theorem_stability} suggests that $t_f\in[0,V(\varepsilon(0))/\kappa]$.
Since $\|\dot{z}_i\|\le\|g_i-g_{i-1}\|\le2$, we have $\|z_i(t_f)-z_i(0)\|\le 2V(\varepsilon(0))/\kappa$.
Therefore, the final converged position $z_i(t_f)$ is sufficiently close to its initial position $z_i(0)$ if the initial angle error $\varepsilon(0)$ is sufficiently small.
In other words, it is impossible that the formation moves through a very long distance given very small initial angle errors.

\section{Simulation Results}\label{section_simulation}
In this section, we present simulation results to illustrate our theoretical analysis.
Figures \ref{fig_simulation_3}, \ref{fig_simulation_4}, \ref{fig_simulation_5} and \ref{fig_simulation_8} respectively show the formation control of three, four, five and eight vehicles.
As shown in the simulations, the proposed control law can efficiently reduce the angle errors and stabilize the formation in finite time.
In our stability proof, we assume that the initial angle error $\varepsilon(0)$ should be sufficiently small such that the initial angle $\theta_i$ and the target angle $\theta_i^*$ are in either $(0,\pi)$ or $(\pi,2\pi)$.
However, as shown in Figures~\ref{fig_simulation_4} and \ref{fig_simulation_8}, even if $\theta_i$ and $\theta_i^*$ may be respectively in the two intervals $(0,\pi)$ and $(\pi,2\pi)$, the formation can still be stabilized.
Hence the simulation suggests that the attractive region of the target formation by the proposed control law is not necessarily small.

\begin{figure}
    \centering
  \subfloat[Vehicle trajectory]{\includegraphics[width=0.5\linewidth]{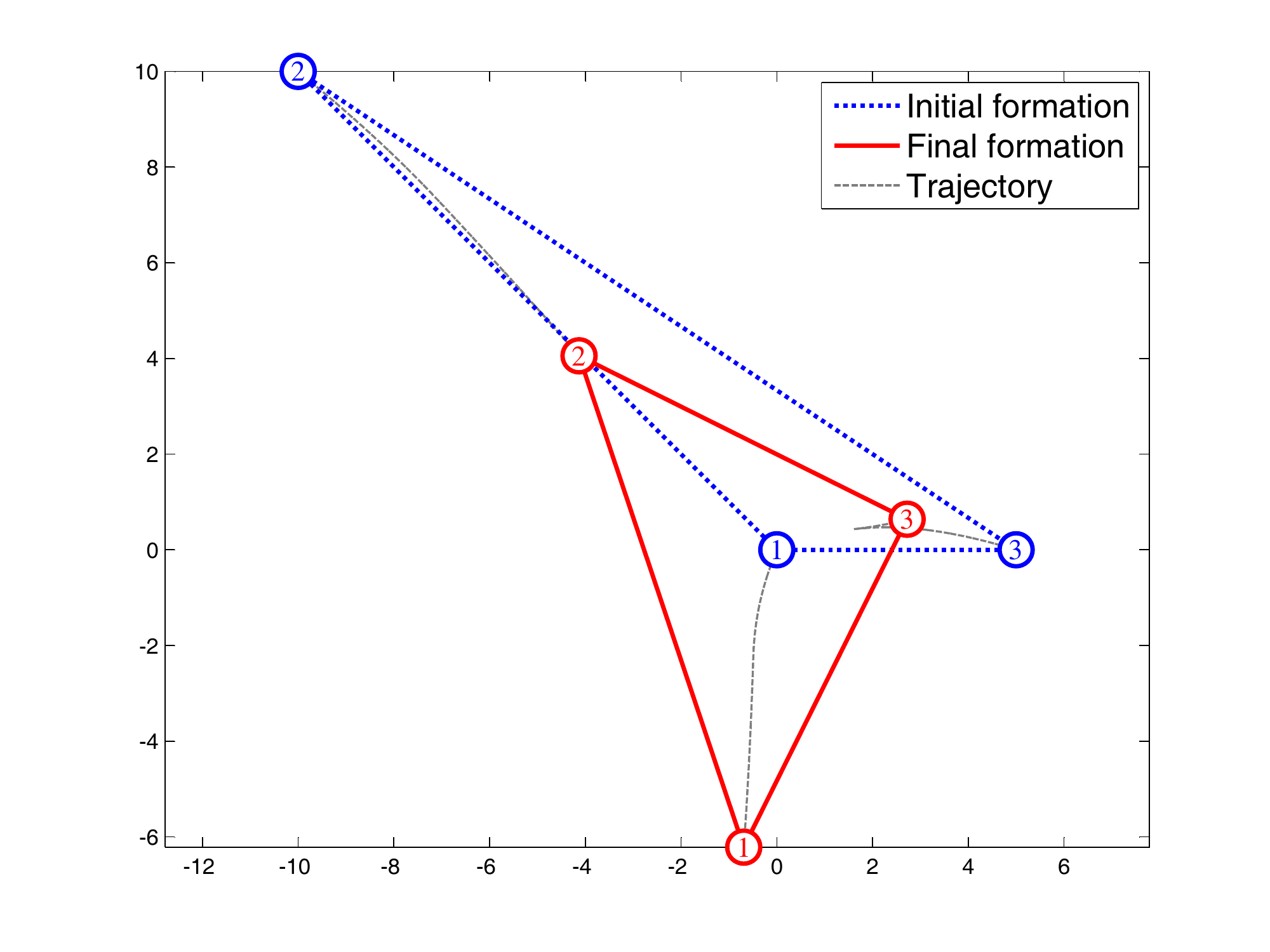}}
  \subfloat[Angle error and Lyapunov function]{\includegraphics[width=0.5\linewidth]{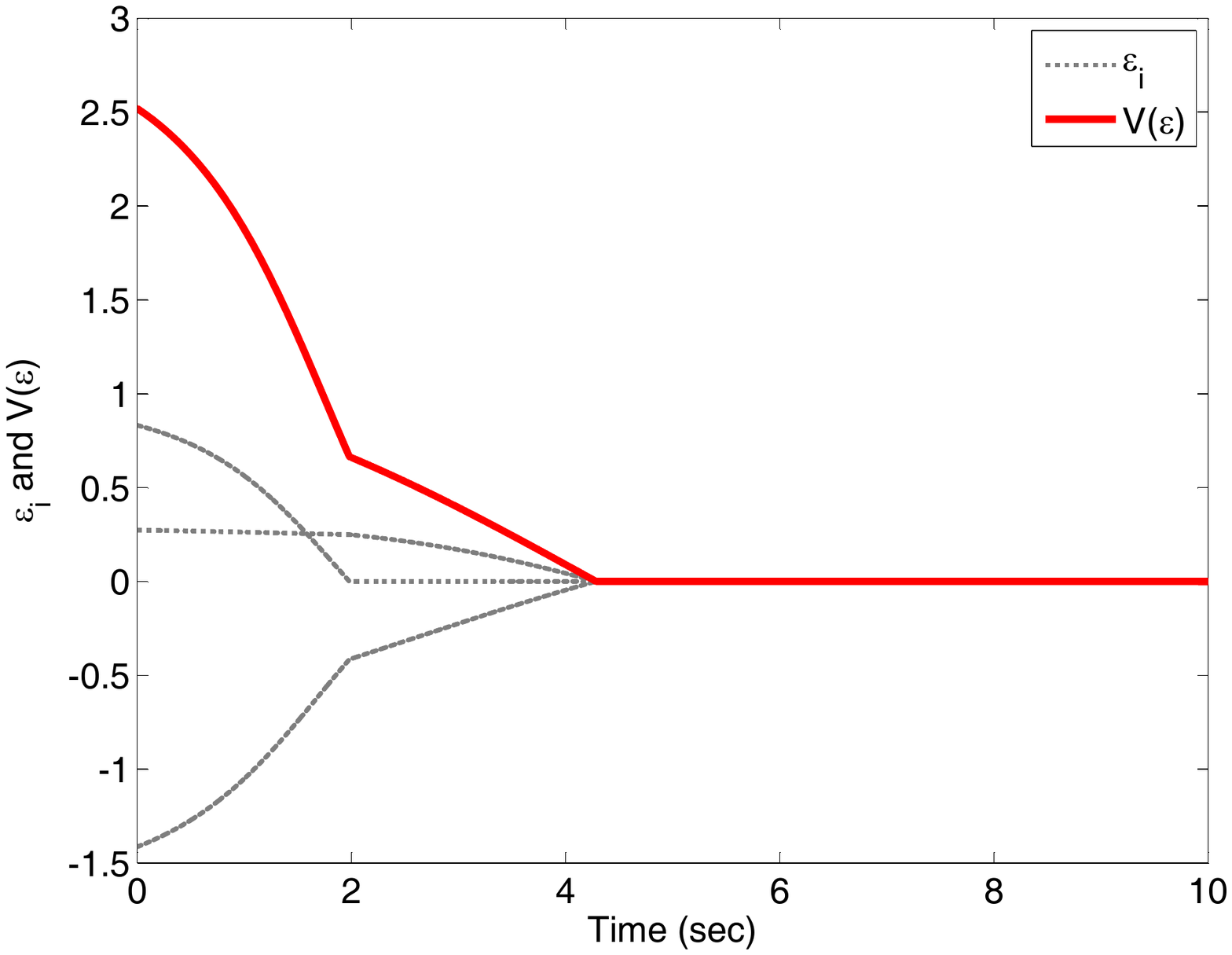}}
  \caption{Control results by the proposed control law with $n=3$, $\theta_1^*=\theta_2^*=45$ deg and $\theta_3^*=90$ deg.}
  \label{fig_simulation_3}
\end{figure}
\begin{figure}
    \centering
  \subfloat[Vehicle trajectory]{\includegraphics[width=0.5\linewidth]{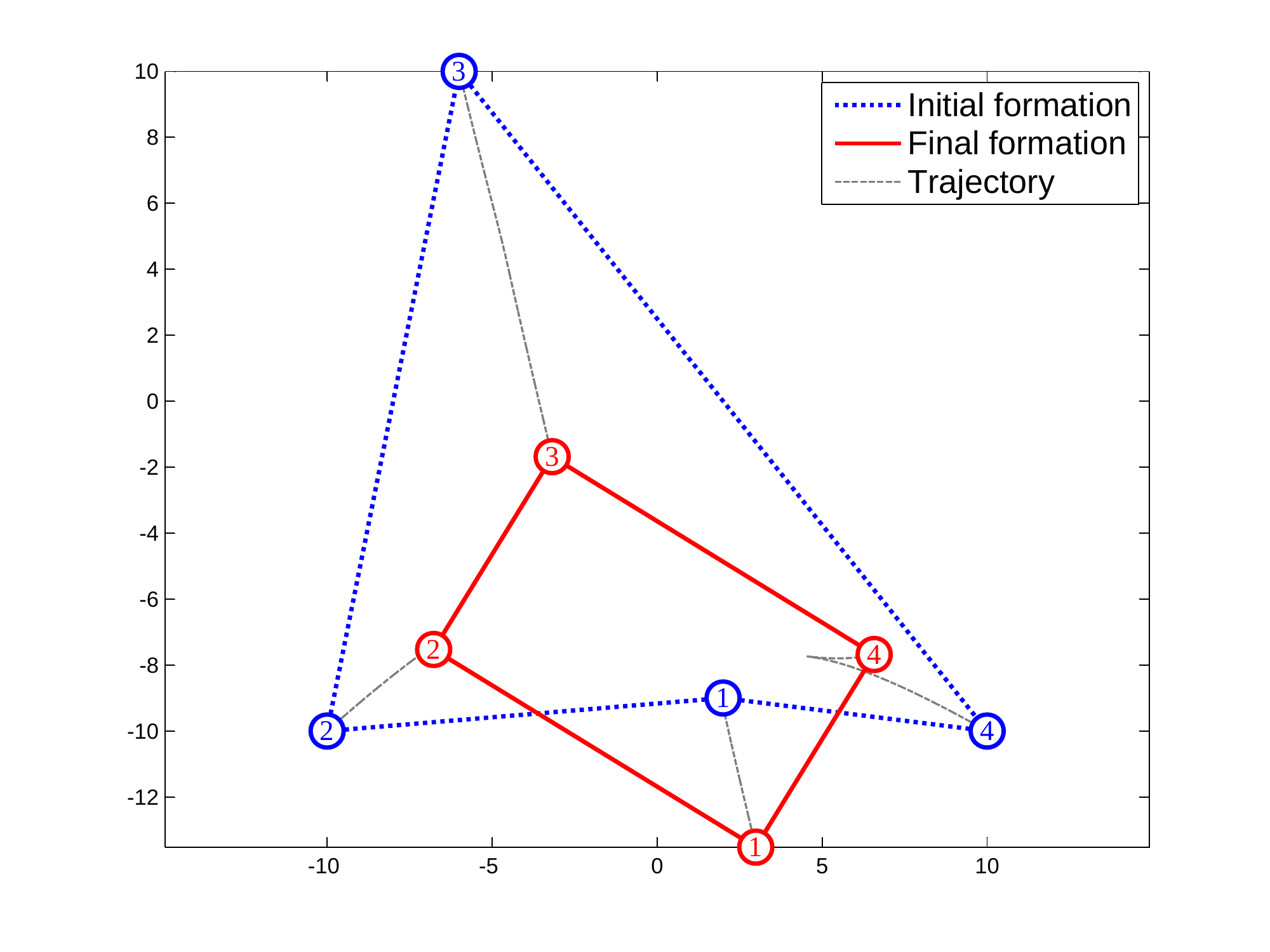}}
  \subfloat[Angle error and Lyapunov function]{\includegraphics[width=0.5\linewidth]{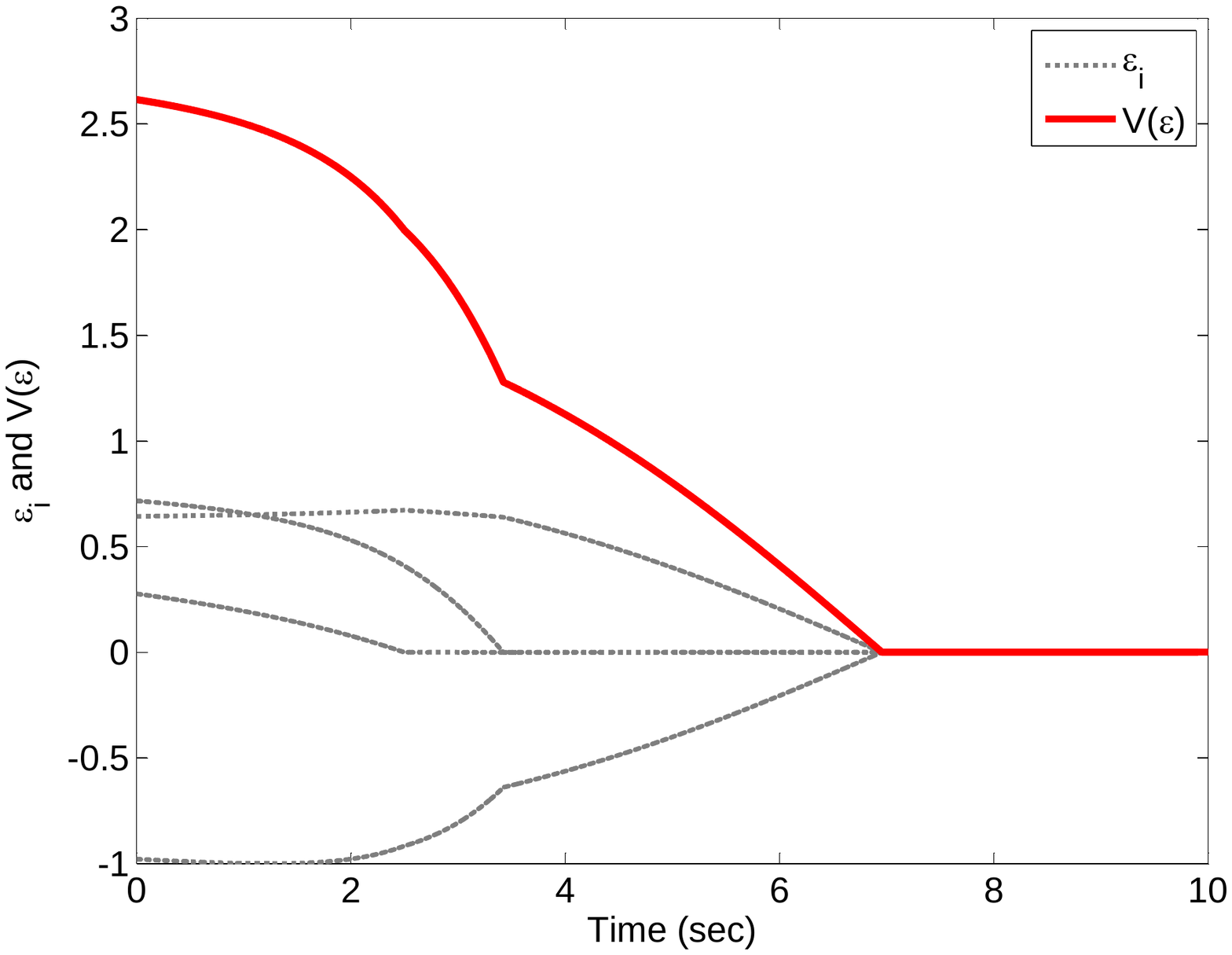}}
  \caption{Control results by the proposed control law with $n=4$ and $\theta_1^*=\dots=\theta_4^*=90$ deg.}
  \label{fig_simulation_4}
\end{figure}
\begin{figure}
    \centering
  \subfloat[Vehicle trajectory]{\includegraphics[width=0.5\linewidth]{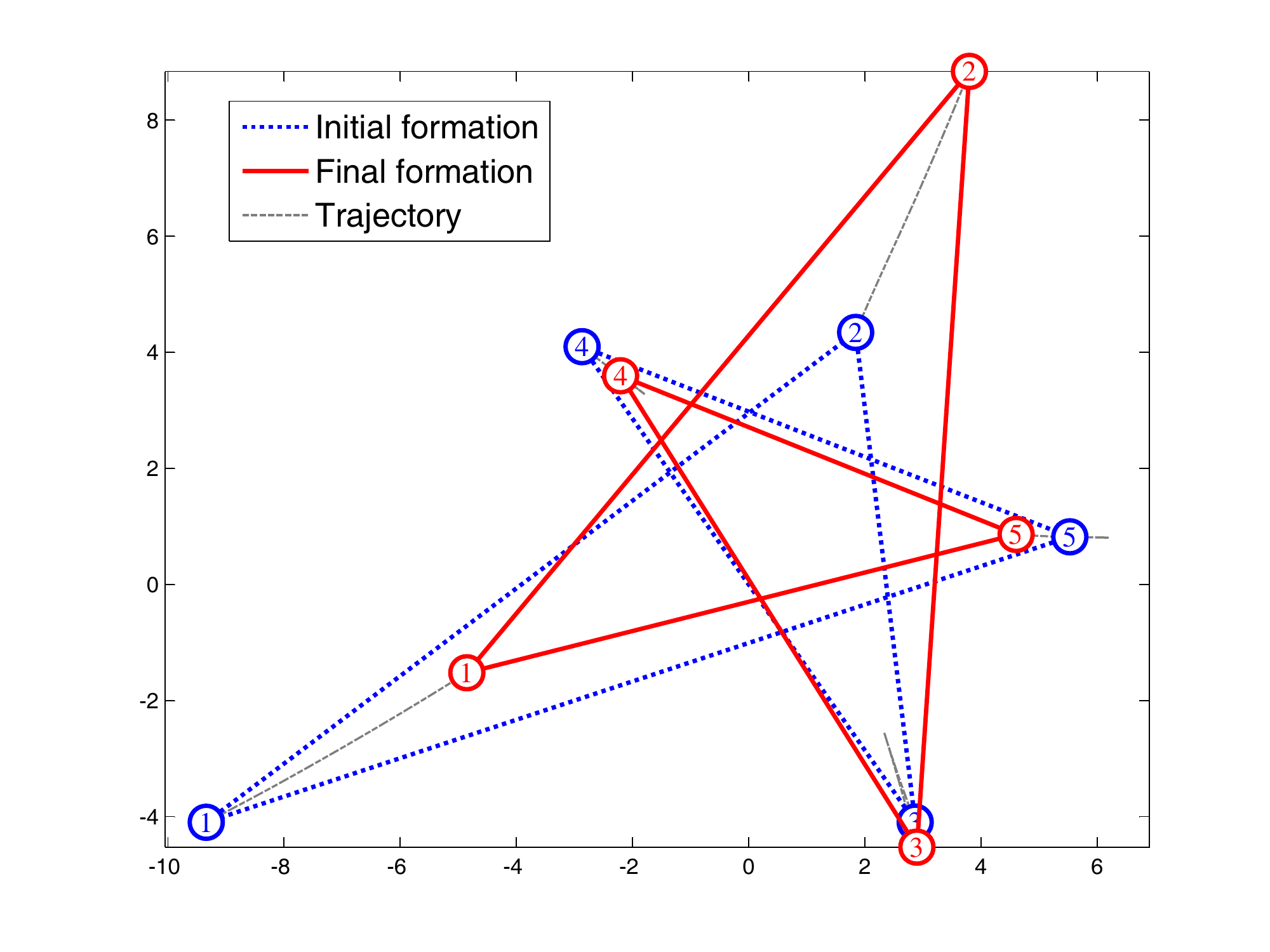}}
  \subfloat[Angle error and Lyapunov function]{\includegraphics[width=0.5\linewidth]{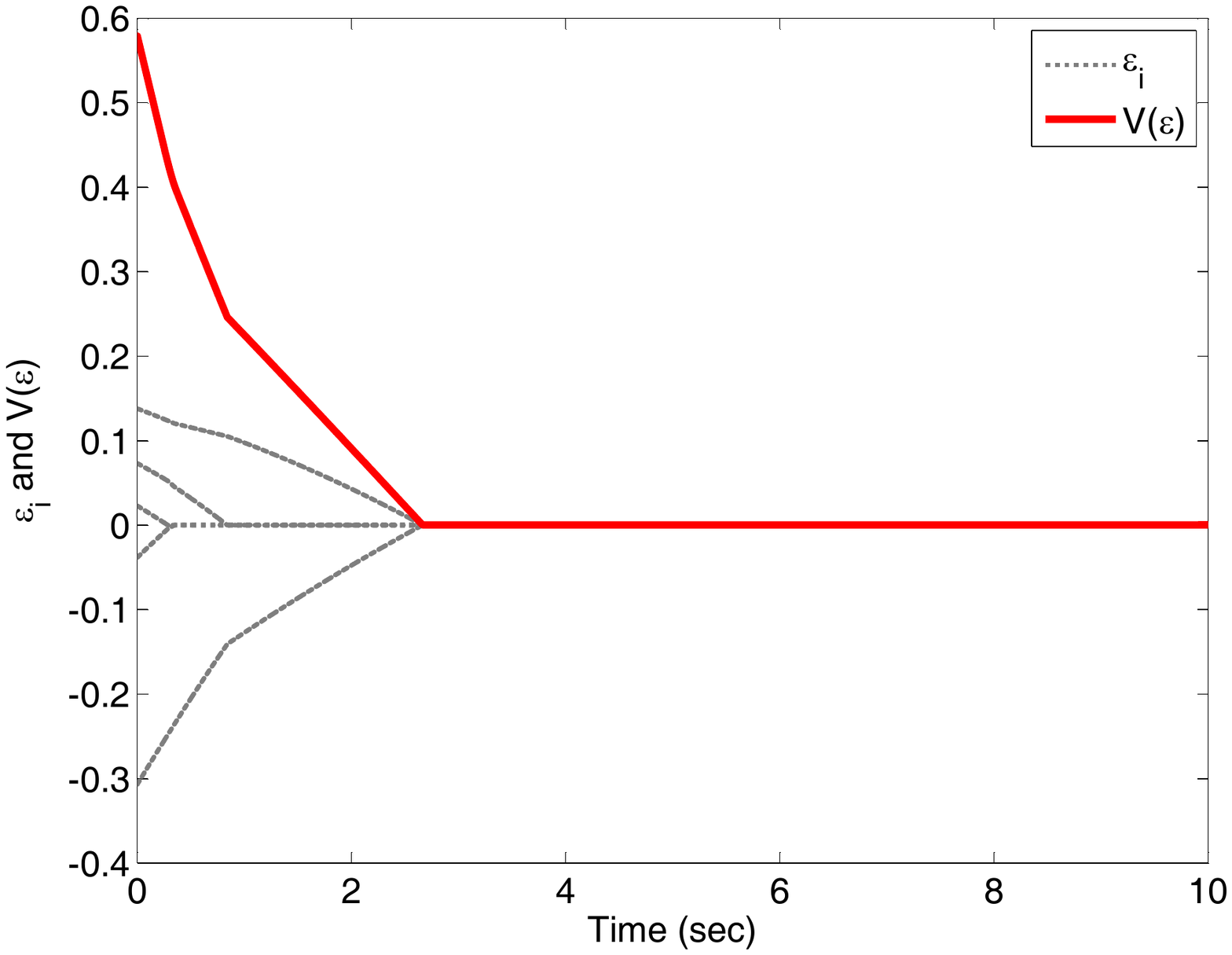}}
  \caption{Control results by the proposed control law with $n=5$ and $\theta_1^*=\dots=\theta_5^*=36$ deg.}
  \label{fig_simulation_5}
\end{figure}
\begin{figure}
    \centering
  \subfloat[Vehicle trajectory]{\includegraphics[width=0.5\linewidth]{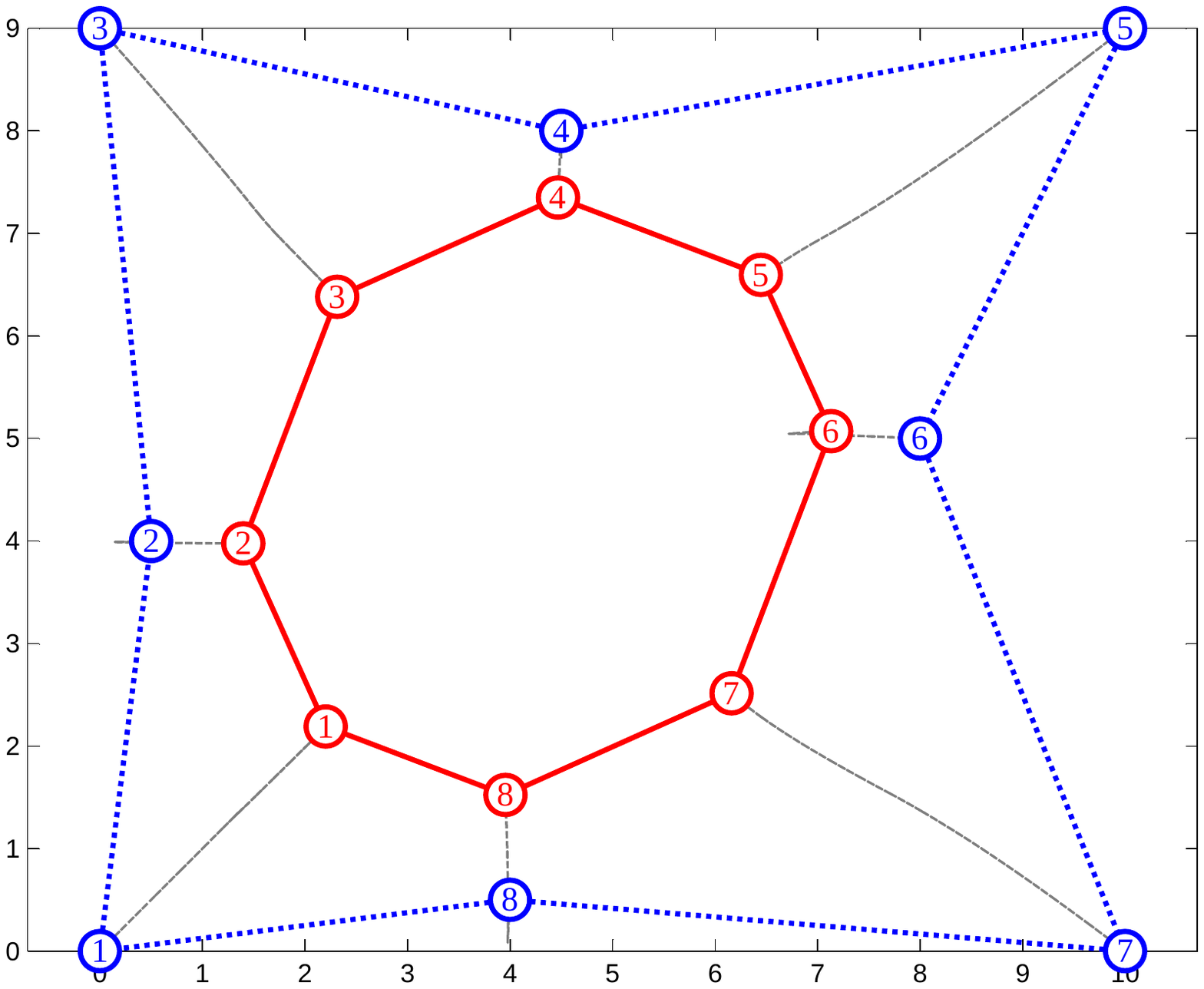}}
  \subfloat[Angle error and Lyapunov function]{\includegraphics[width=0.5\linewidth]{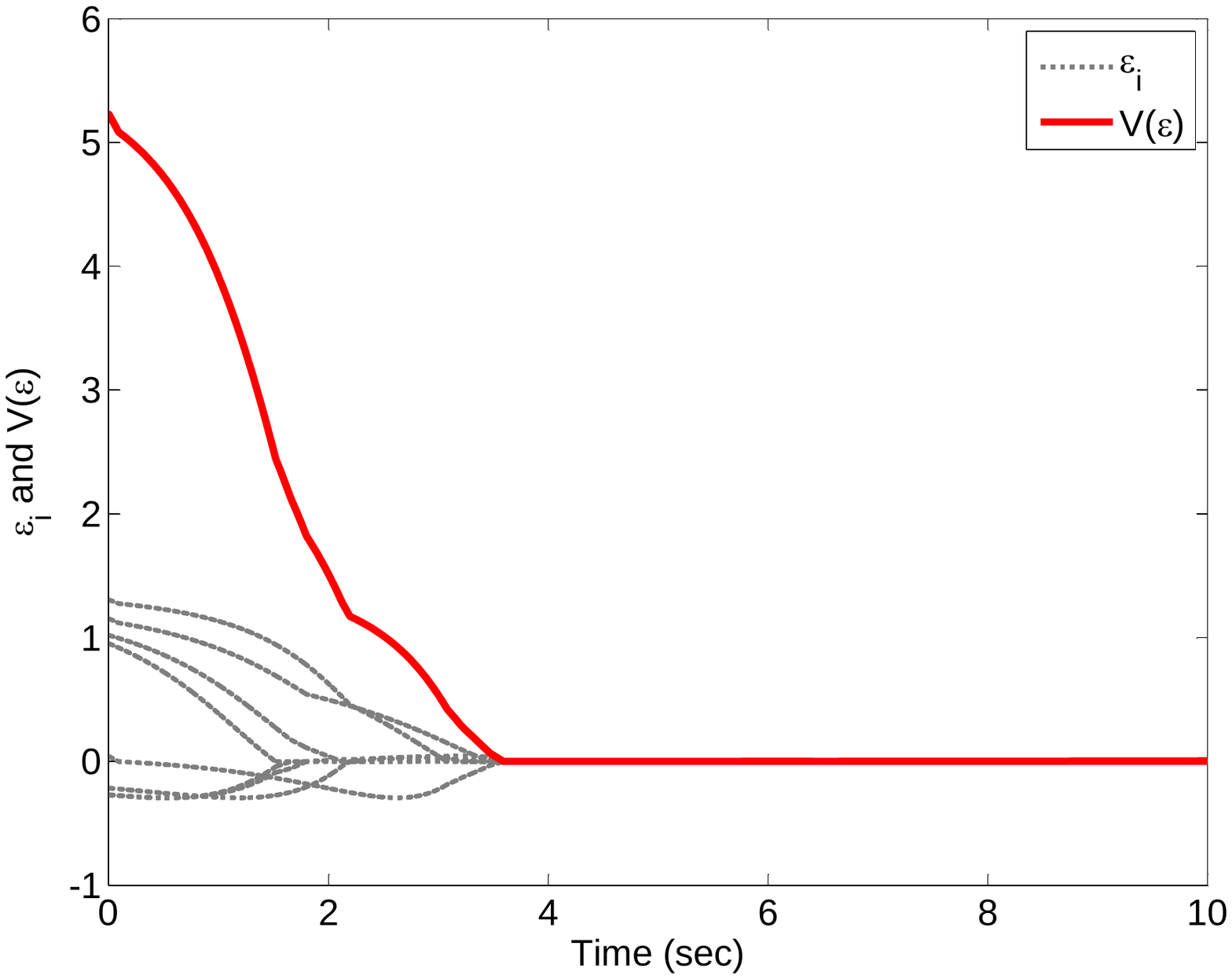}}
  \caption{Control results by the proposed control law with $n=8$ and $\theta_1^*=\dots=\theta_8^*=135$ deg.}
  \label{fig_simulation_8}
\end{figure}

\section{Conclusions}\label{section_conclusion}
We have studied the stabilization of angle-constrained circular formations using bearing-only measurements.
We have proposed a discontinuous control law, which only requires the sign information of the angle errors.
By using nosmooth stability analysis tools, we have proved that the error dynamics is locally finite-time stable with collision avoidance guaranteed.
A number of important formation behaviors have also been identified.

As observed from the simulation results, the shape of the formation cannot be controlled because the underlying information flow is a circular graph, where each vehicle is associated with only one constrained angle.
When there are more than one constrained angles at each vehicle, the formation shape may be well defined.
Then it is possible to control the formation shape using bearing-only measurements.
An immediate research plan is to extend the results in this paper to formations with more complicated underlying graphs.

\bibliography{zsybib}
\bibliographystyle{ieeetr}

\end{document}